\RequirePackage{fix-cm}
\documentclass[smallextended]{svjour3}       
\smartqed  
\usepackage{graphicx}
\usepackage{mathtools}
\usepackage{amssymb}
\usepackage{cite}
\usepackage{pdfsync}

\newcommand{\R}{{\mathbb R}}
\newcommand{\F}{\mathbb F}
\newcommand{\E}{\mathcal E}

\renewcommand{\P}{\mathbb P}
\newcommand\numberthis{\addtocounter{equation}{1}\tag{\theequation}}

\begin{document}

\title{Distribution of the absolute indicator of random Boolean functions
\thanks{This work was partially done while the first author was funded by the Instituto Nacional de Matematica Pura e Aplicada (IMPA), Rio de Janeiro, RJ - Brazil and the foundation of Coordena\c{c}ao de Aperfeic{c}oamento de Pessoal de Nivel Superior (Capes) of the Brazilian ministry of education
}}


\author{Florian Caullery      \and  Fran\c cois Rodier
}


\institute{F. Caullery  \at
              DarkMatter LLC, Abu Dhabi, United Arab Emirates \\
              \email{florian.caullery@darkmatter.ae}           
           \and
           F. Rodier \at
              Institut Math\'ematiques de Marseille, France\\
              \email{francois.rodier@univ-amu.fr} 
}

\date{Received: date / Accepted: date}

\maketitle

\begin{abstract}
The absolute indicator is one of the measures used to determine the resistance offered by a Boolean function when used in the design of a symmetric cryptosystem. It was proposed along with the sum of square indicator to evaluate the quality of the diffusion property of block ciphers and hash functions. While the behaviour of the sum of square of random Boolean functions was already known, what remained was the study of the comportment of the absolute indicator of random Boolean functions.
As an application, we show that the absolute indicator can distinguish a nonrandom binary sequence from a random one.
\end{abstract}

\section{Introduction}
\label{intro}
Let $\F_{2}^n$ be the $n$-dimensional vector space over the finite field of $2$ elements. The Boolean functions are the functions from $\F_{2}^n$ to $\F_2$. They are used in cryptosystems as they are a convenient way to describe S-Boxes. We refer to \cite{Carlet} for a global survey on the cryptographic applications of Boolean functions. There exist several ways to measure the resistance offered by a Boolean function against specific cryptanalysis. Among them, we should mention the \textit{nonlinearity}, which is the Hamming distance of the function to the set of affine functions, the \textit{absolute indicator} and the \textit{sum of squares}, which are usually grouped into the term of Global Avalanche Criterion. The two latter were introduced in \cite{GAC} to measure the capacity of a Boolean function to ensure the propagation property of a cryptosystem. The relations between the absolute indicator and the other cryptographic measures have been extensively studied as well as the distribution of the absolute indicator of certain specific classes of Boolean functions (sometimes under the name of auto-correlation value), see for example \cite{Zhou2013335, 5339812, canteaut2000propagation, Tarannikov2001, Zheng2001}.

Another possible use of these measures was proposed by one of the the authors during a workshop bringing in industrial representatives and academics (see the online report in \cite{SEME}). The problem set out was to determine if a short binary sequence could be pseudo-random. We proposed to see a binary sequence of length $2^n$ as the truth table of a Boolean function, compute its nonlinearity and absolute indicator and compare it to the expected values of random Boolean functions. The idea came from the fact that it was proved by Schmidt in \cite{Schmidt2015}, finalising the work of Rodier \cite{RodierNBF}, Dib \cite{Dib2010, Dib2014} and Lytsin and Shpunt \cite{Shpunt}, that the nonlinearity of random Boolean functions is concentrated around its expected value. Also, the same kind of result exists for the fourth moment of the nonlinearity of random Boolean functions, which is actually the sum of squares  (see \cite{Rodier_onthe}). 

However, there did not exist a study of the distribution of the absolute indicator of random Boolean functions, we fill the gap with our result. 
The difficulty of our case arises from the fact that we are not dealing only with independent random variables. Indeed, by writing the truth table of a Boolean function as a binary sequence, one can see the absolute indicator as the correlation of order 2 between the sequence and its circular shift (or rotation). Hence, we cannot straightforwardly apply estimates on sums of independent variables. We overcome the problem by carefully analysing the dependencies between the random variables. Also, our techniques allow us to keep the proofs simple and only based on combinatorics. 
As an example, we show that  a short binary sequence would be detected to be non random with respect to the absolute indicator  while it would pass the test with nonlinearity.

We begin by defining the absolute indicator of a Boolean function and set the formal frame for the study of the distribution of the absolute indicator.\\

\begin{definition}
Let $B_n=\{f_n:\F_2^n \rightarrow \F_2\}$ be the set of the Boolean functions of $n$ variables and let $f_n \in B_n$. For all $u \in \F_2^n$, write
\[
\Delta_{f_n}(u) := \sum_{x \in \F_2^n} (-1)^{f_n(x)+f_n(x+u)}.
\]
The absolute indicator of $f_n$ is defined by
\[
\Delta(f_n) := \max_{u \in \F_2^n - \{0\}}|\Delta_{f_n}(u)|.
\]
\end{definition}


Our goal is to show the following theorem on the distribution of the absolute indicator of random Boolean functions.
We denote by $\E(X)$ the expectation of a random variable $X$ and by $f_n : \F_2^n \rightarrow \F_2$ a random Boolean function running over the set $B_n$ provided with equiprobability.

\begin{theorem}
The expectation of the absolute indicator has the following limit:
\[
\frac{\E[\Delta(f_n)]}{\sqrt{2^n\log2^n}}\rightarrow 2
\]
as $n\rightarrow\infty$. Moreover,
\[
\P\bigg[\Big|\frac{\Delta(f_n)}{2 \sqrt{2^n \log 2^n}} - 1\Big| > \epsilon \bigg] \rightarrow 0
\]
for all $\epsilon > 0$ as $n\rightarrow\infty$.
\end{theorem}

The strategy to prove this result is based on ideas developed in \cite{KaiPSL} and in \cite{Rodier_onthe} with suitable alteration and on an idea in \cite{Alon2010}. The main difference in our case is that the estimation of the expected value of the absolute indicator involves dealing with non-independent random variables. By separating the dependent and independent parts in the expectation, we are able to apply classical results from martingales theory and then derive the best estimation possible. 
To sum up, we first prove that, for all $\epsilon > 0$,
\[
\P\bigg[\frac{\E[\Delta(f_n)]}{\sqrt{2^n\log2^n}}>2+\epsilon\bigg] \rightarrow 0
\]
as $n\rightarrow \infty$ and then show that, for all $\delta >0$, the set 
\[
N(\delta) = \bigg\{n > 1 : \frac{\E[\Delta(f_n)]}{\sqrt{2^n\log2^n}}< 2 - \delta \bigg\}
\]
is finite.
Moreover, in the last section, we prove that the absolute indicator of a random Boolean function converges almost surely towards $2\sqrt{2^n\log2^n}$. 

This article is an expansion of a paper which was presented at WCC17 \cite{cr}. \\

We begin with preliminary lemmata which shall be used in the fourth section to prove Theorem 1.

\section{Preliminary lemmata}

From now on we set $l = 2^n$. 
\begin{lemma}\label{MercerLike}
For all $\epsilon > 0$, as $n \rightarrow \infty$,
\[
\P\bigg[\frac{\Delta(f_n)}{\sqrt{l\log l}}>2+\epsilon\bigg] \rightarrow 0.
\]
\end{lemma}

\begin{proof}

Write $\mu_l = (2+\epsilon) \sqrt{l \log l}$. The union bound gives:
\begin{align*}
\P(\Delta(f_n)>\mu_l ) & \leq  \sum_{u \in \F_2^n-\{0\}} \P(|\Delta_{f_n}(u)|>\mu_l).
\end{align*}
Now note the trivial fact that $f_n(x)+ f_n(x+u) = f_n(x+u)+f_n(x+u+u)$. Hence, choosing one subset $C_u \subset \F_2^n$ of maximal cardinality such that if $x \in C_u$, then $x+u \not \in C_u$, we can write
\[
\Delta_{f_n}(u) = 2 \sum_{x \in C_u} X_{x,u},
\]
where $X_{x,u} = (-1)^{f_n(x)+ f_n(x+u)}$. Since $f_n$ is drawn at random,  we know from proposition 1.1 of \cite{Mercer} that the $X_{x,u}$'s are independent random variables equally likely to take the value $-1$ or $+1$. We can now apply Corollary A.1.2 of \cite{Alon2010} with  $k = \sharp C_u = l/2$ to obtain
\begin{align*}
\sum_{u \in \F_2^n-\{0\}} \P(|\Delta_{f_n}(u)|>\mu_l) &  = \sum_{u \in \F_2^n-\{0\}} \P\bigg[\Big|\sum_{x \in C_u} X_{x,u}\Big|> \mu_l/2 \bigg]
\\ &  \leq 2 l e^{-\mu_l^2/4l} \\
&<2 l^{- \epsilon} 
\end{align*}
which tends to $0$ as $n \rightarrow \infty$.
\qed\end{proof}

We now prove a lower bound on $\P(\Delta(f)>\lambda )$. We need the following refinement of the central limit theorem.

\begin{lemma}[\cite{Cramer1938}, Thm. 2]\label{Cramer}
Let $X_0, X_1, \ldots$ be i.i.d random variables satisfying $\E[X_0] = 0$ and $\E[X_0^2] = 1$ and suppose that there exists $T>0$ such that $\E[e^{tX_0}]< \infty$ for all $|t|<T$. Write $Y_k=X_0+X_1+\ldots+X_{k-1}$  and let $\Phi$ be the distribution function of a normal random variable with zero mean and unit variance. If $\theta_k>1$ and $\theta_k/k^{1/6}\rightarrow 0$ as $n\rightarrow\infty$, then 
\[
\frac{\P\left[|Y_k| \geq \theta_k \sqrt{k}\right]}{2\Phi(-\theta_k)}\rightarrow 1.
\] 

\end{lemma}

We can now apply lemma 2 to obtain the following proposition.

\begin{proposition}\label{LowerBound}

For all $n$ sufficiently large,
\[
\P \left[ |\Delta_{f_n}(u)|\geq 2 \sqrt{l \log l} \right] \geq \frac{1}{2l \sqrt{ \log l}}.
\]

\end{proposition}

\begin{proof}

From the proof of lemma \ref{MercerLike}, we know that $\Delta_{f_n}(u)$ is the double of a sum of $l/2$ mutually independent variables equally likely to take the value $-1$ or $1$. Notice that $\E[e^{t(-1)^{f_n(0)+f_n(u)}}] = \cosh(t)$ and set $\xi_l = \sqrt{2\log l}$. We can check that $\xi_l / (l/2)^{1/6} \rightarrow 0$ as $n \rightarrow \infty$ and we can now apply lemma \ref{Cramer} to obtain
\[
\P\left[ |\Delta_{f_n}(u)| \geq 2 \sqrt{l \log l} \right] = \P\bigg[ \Big|\sum_{x \in C_u} (-1)^{f_n(x)+f_n(x+u)}\Big| \geq \sqrt{l \log l} \bigg] \sim 2 \Phi(-\xi_l)
\]
with $\Phi$ as in lemma \ref{Cramer}.
Now use the fact that 
\[
\frac{1}{\sqrt{2\pi}z}\left(1 - \frac{1}{z^2} \right)e^{-z^2/2}\leq \Phi(-z) \leq \frac{1}{\sqrt{2\pi}z}e^{-z^2/2} \quad \text{ for } z>0.
\]
So, as $l\rightarrow \infty$,
\[
2 \Phi(- \xi_n) \sim \frac{1}{l\sqrt{\pi  \log l}},
\]
from which the lemma follows.
\qed \end{proof}

\section{Upper bound on $\P \left[ |\Delta_{f_n}(u)|\geq \lambda_l \cap |\Delta_{f_n}(v)|\geq \lambda_l \right]$}

We will proceed by first estimating the expected value of $\left(\Delta_{f_n}(u) \Delta_{f_n}(v)\right)^{2p}$ and then use Markov's inequality. \\

Let $ 0 < r < l$ and choose $r$ elements $x_1, \dots, x_r$ in $\F_{2}^n$, let $f_n$ be a random function in $B_n$, and write $\tilde{f} = (-1)^{f_n}$. First remark that the following properties trivially hold:

\begin{itemize}
	\item $\E[\tilde{f}(x_1) \tilde{f}(x_2) \ldots \tilde{f}(x_r)] = \E[\tilde{f}(x_3) \tilde{f}(x_4) \ldots \tilde{f}(x_r)] \text{ if } x_1 = x_2$
	\item $\E[\tilde{f}(x_1) \tilde{f}(x_2) \ldots \tilde{f}(x_r)] = 0 \text{ or } 1$	
	\item $\E[\tilde{f}(x_1) \tilde{f}(x_2) \ldots \tilde{f}(x_r)] = 1$  if and only if for every $y \in \F_{2}^n$ the set of the $ x_i$'s equals to $ y $ is of even cardinality.

\end{itemize}
Choose $r$  elements $a_1, \ldots, a_r$ of $\F_{2}^n$ and define
\[
E[a_1, \ldots, a_r] = \sum_{(x_1, \dots, x_r) \in \F_2^n \times \F_2^n \times \ldots \times \F_2^n} \E\left[ \tilde{f}(x_1) \tilde{f}(x_1 + a_1)  \ldots \tilde{f}(x_r) \tilde{f}(x_r + a_r)   \right].
\]

\begin{lemma} \label{aplusi}
For $a \in \F_2^n, a \neq 0$, the following inequality is true:
\[
E[a, a_1, \ldots, a_r] \leq 2 \sum_{1 < i \leq r} E[a_1, \ldots, a_i + a, \ldots, a_r].
\]

\end{lemma}

\begin{proof}
(see \cite{Rodier_onthe})
From the previously stated properties, we can write
\begin{align*}
& E[a, a_1, \ldots, a_r] 	\\& = \sum_{(x, x_1, \dots, x_r)} \E\left[ \tilde{f}(x) \tilde{f}(x+a) \tilde{f}(x_1) \tilde{f}(x_1 + a_1)  \ldots \tilde{f}(x_r) \tilde{f}(x_r + a_r)   \right] \\
					& =  \sum_{( x_1, \dots, x_r)} \sum \E\left[ \tilde{f}(x) \tilde{f}(x+a) \tilde{f}(x_1) \tilde{f}(x_1 + a_1)  \ldots \tilde{f}(x_r) \tilde{f}(x_r + a_r)   \right],
\end{align*}
where the last summand is taken over $x \in \{x_1, x_1 + a_1, \ldots, x_r, x_r + a_r \}$. If $x = x_1$,
\begin{align*}
&\E\left[ \tilde{f}(x) \tilde{f}(x+a) \tilde{f}(x_1) \tilde{f}(x_1 + a_1)  \ldots \tilde{f}(x_r) \tilde{f}(x_r + a_r)   \right] \\
&= \E\left[ \tilde{f}(x_1+a) \tilde{f}(x_1 + a_1)  \ldots \tilde{f}(x_r) \tilde{f}(x_r + a_r)   \right]\\
&= \E\left[ \tilde{f}(t) \tilde{f}(t + a_1+a) \tilde{f}(x_1) \tilde{f}(x_1 + a_1)  \ldots \tilde{f}(x_r) \tilde{f}(x_r + a_r)   \right], 
\end{align*}
putting $t = x_1 + a_1$. In the same way, if $ x =  x_1 + a_1$
\begin{align*}
&\E\left[ \tilde{f}(x) \tilde{f}(x+a) \tilde{f}(x_1) \tilde{f}(x_1 + a_1)  \ldots \tilde{f}(x_r) \tilde{f}(x_r + a_r)   \right] \\
&= \E\left[ \tilde{f}(x_1) \tilde{f}(x_1 + a + a_1)  \ldots \tilde{f}(x_r) \tilde{f}(x_r + a_r)   \right].
\end{align*}
\qed
\end{proof}

With this lemma we will show the following.

\begin{lemma}
Let $r$ and $s \geq 2$ be  integers, $a$ and $b$ be distinct elements of $\F_2^n-\{0\}$ and define $S_{r,s} = E[\underbrace{a, \dots, a}_{r}, \underbrace{b, \dots, b}_{s}]$. We have
\[
S_{r,s} \leq 2(s-1)(l + 2r)S_{r,s-2} + 4r(r-1)S_{r-2,s}.
\]
\end{lemma}

\begin{proof}
By successive  application of the inequality stated in lemma \ref{aplusi}, we obtain
\begin{align*}
S_{r,s} 	& \leq  2(s-1) l S_{r,s-2} + 2r E[\underbrace{a + b, a, \dots, a}_{r}, \underbrace{b, \dots, b}_{s}] \\
		&\le 2(s-1)lS_{r,s-2} +4r(r-1)E(\underbrace{b,a,\dots,a,}_{r-1} \underbrace{b,b,\dots,b}_{s-1}) \\
		&\qquad +4r(s-1)E(\underbrace{a,\dots,a,}_{r-1} \underbrace{a,b,\dots,b}_{s-1})\\
		& \leq 2(s-1) (l + 2r) S_{r,s-2} + 4r(r-1) S_{r-2,s}.
\end{align*}
\qed \end{proof}

In the case where $r =0$ and $s$ is even, the lemma gives
\[
S_{0,s} \leq (2 (s-1)l) \E[\underbrace{a, \dots, a}_{s-2}] \leq (2(s-1)l)(2(s-3)l) \ldots ((2.3l)(2l)) \leq l^{s/2} \frac{s!}{(s/2)!}.
\]
We write $M_r = l^{r/2}r! / (r/2)!$,   therefore $S_{0,s}\le M_s$. We get from the preceding Lemma, assuming that $r,s \geq 2$:
\[
\frac{S_{r,s}}{ M_r M_s} \leq (1+rl^{-1}) \frac{S_{r,s-2}}{ M_r M_{s-2}} + 2rl^{-1} \frac{S_{r-2,s}}{ M_{r-2}M_{s}}.
\]

\begin{lemma}

For $2 \leq r \leq s$, $r$ and $s$ both even and putting $t = (r+s)/2$, we get 
\begin{equation}
\label{majoration}
\frac{S_{r,s}}{ M_r M_s} \leq  \left( 1 + \frac{2rs}{l(t-1)} \right)^{(t-1)}.
\end{equation}

\end{lemma}

\begin{proof}

We will proceed by induction. The inequality is clearly true for $r = s = 2$. 
Hence we can suppose  $2 \leq r \leq s$ and $t>2$.
Write $a_{r,s}=1 + \frac{2rs}{l(t-1)}$ with $t=(r+s)/2$.
Suppose now that the relation (\ref{majoration}) is true for 
$\frac{S_{r,s-2}}{ M_r M_{s-2}} $ and $\frac{S_{r-2,s}}{ M_{r-2}M_{s}} $. It implies
\begin{align*}
\frac{S_{r,s}}{ M_r M_s}  	& \leq (1+rl^{-1}) \frac{S_{r,s-2}}{ M_r M_{s-2}} + 2rl^{-1} \frac{S_{r-2,s}}{ M_{r-2}M_{s}} \\
					& = \left( 1 + 2rl^{-1} \right) ( a_{r,s-2})^{(t-2)} + 2rl^{-1}(a_{r-2,s})^{(t-2)} \\
					& =  \left( a_{r,s-2} \right)^{(t-1)} + 2rl^{-1}\left( 1 - \frac{s-2}{t-2} \right) \left( a_{r,s-2} \right)^{(t-2)} + 2rl^{-1} \left(a_{r-2,s}\right)^{(t-2)}.
\end{align*}
 Observe that $1 - \frac{s-2}{t-2} = - \frac{s-r}{2(t-2)}$.
So we want to show that the sum 
\begin{equation}
\label{somme}
\left( a_{r,s} \right)^{(t-1)}   -   \left( a_{r,s-2} \right)^{(t-1)} 
+ 2rl^{-1} \frac{s-r}{2(t-2)} \left( a_{r,s-2} \right)^{(t-2)}  - 2rl^{-1} \left(a_{r-2,s}\right)^{(t-2)}
\end{equation}
is positive.
One has
$
 a_{r,s}   -  a_{r,s-2}  = \frac{2r(s-2)}{l(t-2)(t-1)}  \geq 0,
$
from which we deduce
\begin{align*}
& \left( a_{r,s} \right) ^{t-1} - \left( a_{r,s-2} \right)^{(t-1) } = \frac{2r(r-2)}{ l(t-1)(t-2)} \sum_{i=0}^{t-2} \left(  a_{r,s} \right)^ i \left( a_{r,s-2} \right)^{(t-2-i)}.
\end{align*}
By dividing the  sum (\ref{somme})   by $2r/l$ and using $2 \leq r \leq s$    we get:
\begin{multline*}
 \frac{(r-2)}{ (t-1)(t-2)} \sum_{i=0}^{t-2} \left(  a_{r,s} \right)^ i \left( a_{r,s-2} \right)^{(t-2-i)}
 + \frac{s-r}{2(t-2)} \left( a_{r,s-2} \right) ^{t-2}  -  \left( a_{r-2,s} \right) ^{t-2} 
  \\ \geq  \left( a_{r-2,s} \right) ^{t-2} \left(  \frac{r-2}{t-2} + \frac{s - r}{2(t-2)} - 1  \right). 
\end{multline*}
because 
$
2r(s-2) -2s(r-2) = 4 (s-r) \geq 0\
$
and consequently
$$
a_{r,s} \ge a_{r,s-2}    \geq a_{r-2,s}  .
$$
Finally, we verify that $\frac{r-2}{t-2} + \frac{s-r}{2(t-2)} - 1 = 0$.
\qed 
\end{proof}

Now, we can conclude that 

\[
S_{2p,2p} \leq l^{2p} \left(  \frac{(2p)!}{p!}\right)^2 \left(  1 + \frac{8p^2}{l(2p -1)} \right)^{2p - 1},
\]
which will be used to prove the following:

\begin{proposition} \label{proposition}
Write $\lambda_l = 2\sqrt{l \log l}$. If $f_n$ runs over the set $B_n$ then, for distinct $u, v \in \F_2^n-\{0\}$ and $n$ sufficiently large, 

\[
\P \left[ |\Delta_{f_n}(u)|\geq \lambda_l \cap |\Delta_{f_n}(v)|\geq \lambda_l \right] < 4 l^{-2}.
\]
\end{proposition}

\begin{proof}
With the notation of previous lemmas and applying Markov's inequality:
\begin{align} \label{befStirl}
&\P \bigg( \Big( \sum_{x \in \F_{2}^n} \tilde{f}(x)\tilde{f}(x+u)  \geq \theta_1 \Big) \cap \Big( \sum_{x \in \F_{2}^n} \tilde{f}(x)\tilde{f}(x+v)  \geq \theta_2 \Big) \bigg)  & \nonumber \\
& \leq \E \bigg[ \Big( \sum_{x \in \F_{2}^n} \tilde{f}(x)\tilde{f}(x+u) \sum_{x \in \F_{2}^n} \tilde{f}(x)\tilde{f}(x+v)  \Big)^{2p} \bigg] \Big/ (\theta_1 \theta_2) ^{2p} & \nonumber \\
& \leq S_{2p,2p} / (\theta_1 \theta_2) ^{2p} \leq \left( 1+ \frac{8p^2}{l(2p -1)}  \right)^{2p-1} l^{2p} \left( \frac{(2p)!}{p!} \right)^2 \Big/ (\theta_1 \theta_2) ^{2p} &.
\end{align}
If we take $\theta_1 = \theta_2 = \lambda_l$ and $p =n$, we have
\begin{multline*}
\P \bigg( \Big( \sum_{x \in \F_{2}^n} \tilde{f}(x)\tilde{f}(x+u)   \geq \lambda_l \Big) \cap \Big( \sum_{x \in \F_{2}^n} \tilde{f}(x)\tilde{f}(x+v)  \geq \lambda_l \Big) \bigg) \\
\leq \left( 1+ \frac{8n^2}{l(2n -1)}  \right)^{2n-1} \left( \frac{(2n)!}{n!} \right)^2 \Big/ (4n) ^{2n}.
\end{multline*}
By Stirling's approximation, $\sqrt{2\pi k} k^k e^{-k} \leq k! \leq \sqrt{3 \pi k} k^k e^{-k}$, we get 
\[
\frac{(2n)!}{n!} \leq \frac{\sqrt{3 \pi 2n} (2n)^{2n} e^{-2n}}{\sqrt{2\pi n} n^n e{-n}} \leq \sqrt{3}.2^{2n} n^n e^{-n}.
\]
Plugging it into inequality \ref{befStirl}, we get:
\begin{multline*}
\P \bigg( \Big( \sum_{x \in \F_{2}^n} \tilde{f}(x)\tilde{f}(x+u)   \geq \lambda_l \Big) \cap \Big( \sum_{x \in \F_{2}^n} \tilde{f}(x)\tilde{f}(x+v)  \geq \lambda_l \Big) \bigg)\\
 \leq 3 \left( 1+ \frac{8n^2}{l(2n -1)}   \right)^{2n-1} e^{-2n}< l^{-2}
\end{multline*}
for $n \geq 7$. The proposition is now straightforward.

\qed \end{proof}

\section{Proof of theorem 1}

We first recall the following inequality from martingales theory:

\begin{lemma}[\cite{McDarmid}]\label{McDiarmidInq}
Let $X_0,\ldots, X_{l-1}$ be mutually independent random variables taking values in a set $S$. Let $g:S^l \rightarrow \R$ be a measurable function and suppose that
\[
|g(x)-g(y)|\leq c
\]
whenever $x$ and $y$ differ only in one coordinate. Define the random variable $Y = g(X_0,\ldots, X_{l-1})$. Then, for all $\theta \geq 0$,
\[
\P\left[|Y-\E[Y]|\geq \theta\right] \leq 2 \exp\Big(-\frac{2\theta^2}{c^2l}\Big).
\] 
\end{lemma}

Now let $(x_i)_{0\le i\le l-1}$ be a bijection between the set of nonnegative numbers strictly smaller than $l$ and the set $\F_2^n$.
Let 
$\sigma_u$ for $u\in\F_2^n$ be the substitution in the  set $\{0,\dots,l-1\}$ such that
$x_{\sigma_u( i)}= x_i+u$ and let $g$ be the function
\begin{eqnarray*}
g: \{0,1\}^l &\rightarrow &\R \\
(X_0, \ldots, X_{l-1}) &\mapsto &\max_{u\in\F_2^n-\{0\}} \sum_{i = 0}^{l-1} (-1)^{X_i+X_{\sigma_u(i)}}.
\end{eqnarray*}
 Clearly $g(f(x_0), \ldots, f(x_{l-1}) )= \Delta(f)$ and we can apply lemma
\ref{McDiarmidInq} with $c = 4$ to obtain the following corollary.

\begin{corollary}\label{Corol}

For $\theta \geq 0$,
\[
\P\left[|\Delta(f_n) - \E[\Delta(f_n)] |\geq \theta \right] \leq 2\exp\Big({-\frac{\theta^2}{8l}}\Big).
\] 

\end{corollary}

The first part of Theorem 1 is proven by the following result.

\begin{theorem} \label{firstPart}
The following limit holds when $n\rightarrow \infty$,
\[
\frac{\E\left[\Delta(f_n) \right]}{\sqrt{l \log l}} \rightarrow 2.
\]
\end{theorem}

\begin{proof}

By the union bound and triangle inequality, we have, for all $\epsilon > 0$
\begin{multline*}
\P\left[\frac{\E\left[\Delta(f_n) \right]}{\sqrt{l \log l}} - 2  > \epsilon \right] \\ \leq \P\left[\frac{\E\left[\Delta(f_n) \right]}{\sqrt{l \log l}} - \frac{\Delta(f_n)}{\sqrt{l \log l}} > \frac{1}{2}\epsilon \right] + \P\left[ \frac{\Delta(f_n)}{\sqrt{l \log l}} - 2 > \frac{1}{2}\epsilon \right] .
\end{multline*}
The right hand side of the last inequality goes to zero as $n \rightarrow \infty$ by corollary \ref{Corol} and proposition \ref{MercerLike}. So we conclude
\[
\limsup_{n \rightarrow \infty} \frac{\E\left[\Delta(f_n) \right]}{\sqrt{l \log l}} \leq 2.
\] 

The proof of the claim is based on an idea in \cite{Shpunt}: to bound by below $\frac{\E\left[\Delta(f_n) \right]}{\sqrt{l \log l}} $, we will prove that the following set is finite. Let $\delta > 0$ and define
\[
N(\delta) = \left\lbrace n>1: \frac{\E\left[\Delta(f_n) \right]}{\sqrt{l \log l}} < 2- \delta \right\rbrace.
\]
Now set $\lambda_l = 2\sqrt{l \log l}$ and, for each $n \geq7$, chose a subset $W \subset \F_2^n-\{0\}$ of size $\lceil l/\log l \rceil$.  Hence 
\begin{multline*}
\P\left[\Delta(f_n) \geq \lambda_l \right]  \geq  \P\Big[\max_{u \in W} |\Delta_{f_n}(u)| \geq \lambda_l \Big] \\
 \geq \sum_{u \in W} \P\left[|\Delta_{f_n}(u)| \geq \lambda_l \right] - \sum_{u,v \in W, u \not = v} \P\left[|\Delta_{f_n}(u)|\geq \lambda_l \cap |\Delta_{f_n}(v)|\geq \lambda_l \right]
\end{multline*}
by Bonferroni inequality. Propositions \ref{LowerBound} and \ref{proposition} give, for $l$ big enough:
\begin{align*}
\P\left[\Delta(f_n) \geq \lambda_l \right] & \geq |W| \frac{1}{2l\sqrt{ \log l}} - 4\frac{|W|^2}{l^2}\\
 & \geq \frac{1}{10(\log l)^{3/2}}. \numberthis \label{eqn}
\end{align*}
By definition of $N(\delta)$, $\lambda_l \geq \E[\Delta(f_n)]$ so we can apply corollary \ref{Corol} with $\lambda_l - \E[\Delta(f_n)]$ so that for all $n \in N(\delta)$, 
\[
\P\left[ \Delta(f_n) \geq \lambda_l \right] \leq 2\exp\Big({-\frac{1}{8l}(\lambda_l - \E[\Delta(f_n)])^2}\Big).
\]
Comparison with \eqref{eqn} implies
\[
\frac{\E[\Delta(f_n)]}{\sqrt{l \log l}} \geq 2 - \sqrt{\frac{12 \log \log l + 8 \log 20}{\log l}},
\]
which means in view of its definition that $N(\delta)$ is finite for all $\delta > 0$.
\qed \end{proof}

We can now easily deduce the second part of Theorem 1.

\begin{corollary}

As $n \rightarrow \infty$,
\[
\frac{\Delta(f_n)}{\sqrt{l \log l}} \rightarrow 2 \quad \textit{in probability.}
\]

\end{corollary}

\begin{proof}
By the triangle inequality 
\begin{multline}\label{triangularInequality}
\P \left[ \left| \frac{\Delta(f_n)}{\sqrt{l \log l}} - 2 \right| > \epsilon \right] \\ \leq \P \left[ \left| \frac{\Delta(f_n)}{\sqrt{l \log l}} - \frac{\E\left[ \Delta(f_n) \right]}{\sqrt{l \log l}} \right| > \frac{\epsilon}{2} \right] + \P \left[ \left| \frac{\E\left[ \Delta(f_n) \right]}{\sqrt{l \log l}} - 2 \right| > \frac{\epsilon}{2} \right]. 
\end{multline}

By Theorem \ref{firstPart} and Corollary \ref{Corol}, the two terms on the right-hand side go to 0 as $n \rightarrow \infty$ which proves the corollary.

\qed 
\end{proof}

The combination of these two last theorems gives the proof of Theorem 1.

\section{Stronger convergence}

The goal of this section is to use the Borel-Cantelli Lemma to prove the following Theorem: 

\begin{theorem}
Denote by $\Omega$ the set of infinite sequences of elements of $\F_2$ and by $B$ the space of functions from $\Omega$ to $\F_2$. For every $f \in B$, we denote its restriction to its $n$ first coordinates by $f_n$ which is in $B_n$. We define the following probability measure on $B$:
\[
\P[f\in B: f_n = g \in B_n] = 2^{-2^n}
\]
for all $g \in B_n$. Now let $f$ be a random element in $B$. Then
\[\
\lim_{n\rightarrow \infty} \frac{\Delta(f_n)}{\sqrt{l \log l}} = 2 \text{ almost surely.}
\]
\end{theorem}

\begin{proof}
Let $f$ be a random element in $B$ and $f_n$ its restriction. We know from the proof of Theorem 1 that 
the inequality (\ref{triangularInequality}) is true.
By Theorem 2, the second term of the right hand side in (\ref{triangularInequality}) is null for $n$ sufficiently large. From Corollary \ref{Corol}, the first term is bounded by $2l^{-\frac{\epsilon^2}{32}} = 2^{\frac{-n\epsilon^2 + 32}{32}}$.  Therefore, for $\epsilon > 0$:

\[
\sum_{n = 1}^{\infty} \P\bigg[\Big|\frac{ \Delta(f_n)}{\sqrt{l \log l}} - 2 \Big| > \epsilon  \bigg] < \infty.
\]
Now applying the Borel-Cantelli Lemma gives the theorem.
\qed
\end{proof}

\section{The non-linearity vs the autocorrelation}

\subsection{Application: a test on  Boolean functions}

The article [CGM + 14] proposes as application of the absolute indicator a test on the Boolean functions to know if a random Boolean function (for example given by a True Random Number Generator) would be reliable (see the Introduction).
We will show here that autocorrelation can sometimes detect a function that cannot pass the test whereas with non-linearity it does.

Indeed a common weakness in random function generators is that a function can always loop, that is, it can be periodic.
We will see that a function with two periods can pass the test of non-linearity, but not that of autocorrelation.

\subsection{Two periods function}

Let a Boolean function on $\F_2^n$ be defined by
$$\begin{array}{rccl}
f: \F_2^n=  &\F_2^{n-1}\times\F_2 &\longrightarrow &\F_2 \cr
&(x,0)& \longmapsto &g (x) \cr
&(x,1)& \longmapsto &g (x).
\end{array}$$
where $g$ is  a Boolean function on $\F_2^{n-1}$.

\subsubsection{The nonlinearity} 

We want to compute the nonlinearity of $f$ based on that of $ g $.

\begin{proposition}
The nonlinearity of $f$ (computed in $\F_2^n$) is the double of the nonlinearity of  $ g $ (computed in $\F_2^{n-1}$).
\end{proposition}

\begin{proof}

Using the formulas
$$NL (f) = 2^{n-1} -1/2 \max_{u\in\F_2^n} |\widehat f (u)|$$
and
$$\widehat f (u) = \sum_{x\in\F_2^n} (-1)^{f (x)+x.u}$$
and by taking $u=(v, a)$ for $ v\in\F_2^{n-1}$ and $a\in\F_2 $,
it is an easy matter to show that
$\widehat f (v, 0)= 2\widehat g (v) \hbox{ , }\widehat f (v,1) = 0$.
Therefore $ NL (f) =2 NL (g).$
\qed\end{proof}

To perform the test we need to compare this nonlinearity to the one of random Boolean functions.
Using the calculation of Litsyn and Shpunt  \cite{Shpunt}, we see that  with a probability of $O(1/n^4)$  most of the function  with $n$ variables will lie in 
 $[lowNL_n,\  highNL_n]$ with
$lowNL_n=\scriptstyle 2^{n-1}-\sqrt{2^{n-1}(n\log 2 +3.5 \log(n\log 2) +0.125)}$ and $highNL_n=\scriptstyle 2^{n-1}-\sqrt{2^{n-1}(n\log2 -4.5 \log(n\log 2))}.$
  As we want to determine the pseudo-randomness of only short sequences (say $n<20$) we have
  $$2lowNL_{n-1}<  lowNL_{n}< 2highNL_{n-1} < highNL_{n} .$$
 This means that the segment supporting the nonlinearity of most functions with a  non-linearity close to $2NL(g)$ intersect  the segment supporting the nonlinearity of most random Boolean functions with $n$ variables.  So the function $f$ is likely to pass the test.

\subsubsection{The autocorrelation} 

\begin{proposition}
The autocorrelation of $f$ is $2^n$.
\end{proposition}

\begin{proof}
If we  take $u=(0,\dots,0,1)=({{\bf 0}},1)\in$ $\F_2^{n-1}\times\F_2$ it is straigthforward to check that 
$$\Delta(f)= \Delta_{f}(u)  = 2^n.$$
\qed
\end{proof}

This value is quite different from the value expected for the autocorrelation of most of the functions which is
$2\sqrt{2^n\log(2^n)}$. 
Hence the function $f$ will not pass the test with autocorrelation, as with nonlinearity it may.

\subsection{Disturbed two period function}
We want to stress the fact that this phenomenon still exist if the function is slightly disturbed (by a noise for instance).

Define inductively a sequence $(f_i)_{0\le i \le r}$ such that $f_0=f$, and $f_{i+1}$ is  is obtained from $f_i$ by  choosing a random value $i$ in the truth table of $f_i$ and changing the sign of $f_i(u_i)$.
So  the Boolean function $f_r$  is equal to $f$, except for a set $E$ of $r$ points.

\subsubsection{The autocorrelation}

\begin{proposition}
The autocorrelation of $f_r$ fulfills
$$\Delta(f_r)  \ge 2^{n} -4r$$
where $ r $ is the number of errors in $f_r$ with respect to $ f $.

\end{proposition}

\begin{proof}
Let us take $u=(0,\dots,0,1)=({\bf 0},1)\in\F_2^{n-1}\times\F_2$. Then

\begin{eqnarray*}
\Delta_{f_r}({\bf 0},1)
&=&\sum_{x\in\F_2^n,x\notin E,x+u\notin E} (-1)^{f_r(x)+f_r(x+u)} +\sum_{x\in\F_2^n,x\in E\cup (E+u)} (-1)^{f_r(x)+f_r(x+u)}  \\
&\ge&\sum_{x\in\F_2^n,x\notin E,x+u\notin E} (-1)^{f(x)+f(x+u)} -2r \\
&\ge&2^{n} -4r 
\end{eqnarray*}
because $f_r(x+u)=f_r(x)$ if $x\in\F_2^n,x\notin E,x+u\notin E$.
Therefore
$$\Delta(f_r) = \sup_{u} \Delta_{f_r}(u) \ge  \Delta_{f_r}({\bf 0},1) \ge 2^{n} -4r.$$
\qed
\end{proof}

\subsubsection{The nonlinerarity} 

\begin{proposition}
The nonlinearity of $f_r$ is such that
$$ NL(f)-r \le NL(f_r)\le   NL(f)+r. $$
\end{proposition}
\begin{proof}
The nonlinearity of a Boolean function is the minimum of the number of bits that you need  to change in the truth table of this function to get an affine function.
To get $f_r$ from $f$  you have to change at most $r$ bits. So
$NL(f) \le NL(f_r)+r.$
To get $f$ from $f_r$  you have to change at most $r$ bits. So
$NL(f_r) \le NL(f)+r.$
The conclusion follows.
\qed
\end{proof}

One can even get a more precise estimate of the nonlinearity.
 \begin{proposition}
The following inequality is true:
$$P(|NL(f_r) )-  NL(f)| > s )  \le 2 e^{-s^2/2r}.$$
\end{proposition}
\begin{proof}
Let we define a random variable $X_i$ for $1\le i\le r$ by 
$X_i=NL(f_{i})-NL(f_{i-1})$ so that $NL(f_r) -  NL(f) =\sum_1^r X_i$.
These perturbation are independent, so   Chernoff bound gives the following result:
$$P\bigg[\Big|\sum_1^r X_i\Big|>s \bigg] < 2 e^{-s^2/2r}.$$
Whence the result.
\qed
\end{proof}


This proposition shows that the nonlinearity of $f_r$ does not   deviate too much from the nonlinearity of $f$. Hence again a function slightly disturbed will pass the test for nonlinearity, as it will not for autocorrelation.

\section{Conclusion}

We proved that the absolute indicator of most of the Boolean functions is close to a small value. Thus we draw the  following conclusions: 
\begin{itemize}
\item One should not consider the absolute indicator of a Boolean function as a primary criterion in the design of symmetric cryptographical primitives but focus on other properties relevant to his desired application. The attention should be also given to, for example, simplicity of the algebraic expression (to allow easy bitsliced or efficient hardware implementation) or nonlinearity. Once a Boolean function is selected, the designer should only verify that the absolute indicator is not too far from the expected value.
\item Regarding the application of the absolute indicator proposed in \cite{SEME}, one can say that this test would have an clearly favour type I errors against type II errors, i.e., a string not passing the test is certainly not random while we cannot guaranty that a string passing the test is ``truly'' random. 
\item As an example of this test, we have shown that  some short binary sequence would be detected to be non random with the absolute indicator  while it would pass the test with nonlinearity.

\end{itemize}




\end{document}